\documentclass[twocolumn]{svjour3} 
\smartqed  			
\usepackage{graphicx}
\usepackage{mathptmx}    
\usepackage{amsmath}
\usepackage{amssymb}
\usepackage{graphics}
\usepackage{graphicx}
\usepackage{lmodern,bm}
\usepackage{algorithmic}
\usepackage{booktabs}
\newtheorem{algorithm}{Algorithm}
\newtheorem{comment}[theorem]{Comment}
\newcommand{\hide}[1]{}

\begin{document}
\title{Improved Algorithms for Exact and Approximate Boolean Matrix Decomposition\thanks{This work
was supported by the Ministry of Education of Japan under Scientific Research grants 25280121
and 26560134 awarded to Yuan Sun.} 
}
\titlerunning{Boolean Matrix Decomposition} 
\author{Yuan Sun       
	\and Shiwei Ye
	\and Yi Sun
	\and Tsunehiko Kameda
}

\authorrunning{Yuan Sun, et al.} 
\institute{Yuan Sun \at
              Information and Society Research Division, National Institute of Informatics, Tokyo, Japan\\
              \email{yuan@nii.ac.jp} 
           \and
           Shiwei Ye \at
              School of Electrical \& Communic. Eng., University of China Academy of Science, Beijing, China\\
              \email{shiwye@ucas.ac.cn}
           \and
           Yi Sun \at
              School of Electrical \& Communic. Eng., University of China Academy of Science, Beijing, China\\
              \email{sunyi@ucas.ac.cn}
           \and
           Tsunehiko Kameda \at
           School of Computing Science, Simon Fraser University, Vancouver, Canada  \\           
	\email{tiko@sfu.ca} 
}


\maketitle
\begin{abstract}
An arbitrary $m\times n$ Boolean matrix $M$ can be decomposed {\em exactly}
as $M =U\circ V$,
where $U$ (resp. $V$) is an $m\times k$ (resp. $k\times n$) Boolean matrix
and $\circ$ denotes the Boolean matrix multiplication operator.
We first prove an exact formula for the Boolean matrix $J$ such that $M =M\circ J^T$ holds,
where $J$ is maximal in the sense that if any 0 element in $J$ is changed to a 1
then this equality no longer holds.
Since minimizing $k$ is NP-hard,
we propose two heuristic algorithms for finding suboptimal but good decomposition. 
We measure the performance (in minimizing $k$) of our algorithms on several real datasets
in comparison with other representative heuristic algorithms for Boolean matrix decomposition (BMD).
The results on some popular benchmark datasets demonstrate that one of our proposed algorithms
performs as well or better on most of them.
Our algorithms have a number of other advantages: 
They are based on exact mathematical formula,
which can be interpreted intuitively.
They can be used for approximation as well with competitive ``coverage.''
Last but not least, they also run very fast.
Due to interpretability issues in data mining,
we impose the condition, called the ``column use condition,''
that the columns of the factor matrix $U$ must form a subset of the columns of $M$. 

In educational databases, the ``ideal item response matrix'' $R$, the ``knowledge state matrix'' $A$
and the ``Q-matrix'' $Q$ play important roles.
As they are related exactly by $\overline{R}=\overline{A}\circ Q^T$, given $R$,
we can find $A$ and $Q$ with a small number ($k$) of ``knowledge states,''  using our exact BMD heuristics.
\end{abstract}

\hide{
}

\section{Introduction}
Matrix decomposition, also called {\em matrix factorization},
has a long history and is an indispensable tool in Matrix algebra~\cite{golub1996}.
Many applications of matrix decomposition to data mining are described in a recent book
on massive data mining by Rajaraman et al.~\cite{rajaraman2014}.
The well-known {\em singular value decomposition} (SVD), for example,
is now a well-established technique,
and has been applied in diverse areas, ranging from statistics, image processing,
signal processing, and data analytics, to name a few.
Although SVD provides a powerful tool in many applications,
it suffers from the lack of interpretability in some applications~\cite{miettinen2008a}.
To address the interpretability issue,
researchers investigated {\em non-negative factorization} (NMF)~\cite{berry2007,lee1999,lee2001,vavasis2010}.
In applications such as digital image analysis, DNA analysis, and chemical spectral analysis,
for example,
it is required that the factor matrices have only non-negative elements.

To deal with categorical data in data mining, 
there have recently been intensive research activities in
{\em Boolean matrix decomposition} ({\em BMD}).
A good overview can be found in the Ph.D. thesis of Miettinen~\cite{miettinen2009},
who laid a ground work on BMD.
In connection with data mining,
BMD has attracted a great deal of research interest in recent years,
as evidenced by a large number of recent publications.
The seminal work by Miettinen et al.~\cite{miettinen2009,miettinen2008b} was a catalyst
to ignite a wave of interest in BMD and its applications to data mining,
for example,
\cite{barnes2010,belohlavek2013,belohlavek2010,koedinger2012,miettinen2008b,miettinen2012,xiang2011},
although there had been some related work prior to that,
e.g., \cite{monson1995,schmidt2011}, in combinatorics research.
BMD also has applications in such areas as educational testing~\cite{tatsuoka2009}
and access control~\cite{streich2009}, 
as well as in more traditional data analysis.

By $M\in \{0,1\}^{m\times n}$, we mean that $M$ is an $m\times n$ Boolean matrix.
BMD aims to find two matrices $U\in \{0,1\}^{m\times k}$ and $V\in \{0,1\}^{k\times n}$
such that the difference $||M-U\circ V||_L$ under some norm
$L$ is minimized with a given $k$ or as small a $k$ as possible.
The minimum possible $k$ is called the {\em Boolean rank} of $M$.
It is known that the Boolean rank of a binary matrix may be larger or smaller
than its real rank~\cite{gregory1983}.
Moreover, the rank of any real matrix can be computed efficiently by Gaussian elimination,
while finding the Boolean rank of a binary matrix is NP-hard~\cite{nau1978}.

In this paper,
we initially require that $||M-U\circ V||_L =0$ under any norm $L$,
namely we are interested in {\em exact} BMD.\footnote{Later in this paper we relax
the requirement of exact decomposition,
and also discuss approximation to BMD.}
Therefore, unless otherwise specified, $||M||$ (resp. $||\bm{v}||$) shall
denote the number of non-0 elements in $M$ (resp. vector $\bm{v}$),
i.e., the $l_0$ norm.
It is clear that this problem is equivalent to the covering of a bipartite graph with {\em bi-cliques},
as pointed out by Lubiw~\cite{lubiw1990}.
Unfortunately,  the minimum bi-clique covering of a bipartite graph,
hence BMD, is an NP-hard problem~\cite{orlin1977} even for the chordal
bipartite graphs~\cite{mueller1996}.
Therefore it is impractical to insist that we discover $U$ and $V$ with the minimum $k$,
especially when the size of $M$ is large.
For more information on bicliques,
the reader is referred to \cite{amilhastre1998,doherty1999,hochbaum1998}.
It is known that a minimum bi-clique cover can be found in polynomial time
for some subclasses of bipartite graphs~\cite{franzblau1984,lubiw1990,lubiw1991,mueller1996}.

Geerts et al.~\cite{geerts2004} formulate the problem as follows.
A {\em tile} consists of a set of 1's in a Boolean matrix that appear at every intersection
of a set of rows and a set of columns,
and the number of those 1's is called the {\em area} of the tile.
A tile is also called a {\em combinatorial rectangle} in a communications context~\cite{kushilevitz1996}. 
A set of tiles is called a {\em tiling}.
Geerts et al.~\cite{geerts2004} investigate several tiling problems cast in the context of databases.
We paraphrase some of them as problems of covering 1's in a given matrix $M$.
\begin{enumerate}
\item
{\em Minimum tiling}.
Find a tiling containing the smallest number of tiles that together cover all the 1's in $M$.
\item
{\em Maximum $k$-tiling}.
Find a tiling consisting of at most $k$ tiles covering the largest number of 1's in $M$.
\item
 {\em Large tile mining (LTM)}.
 Given a minimum threshold $\sigma$, find all tiles whose area is at least $\sigma$.
\end{enumerate}

Our main goal is to solve the minimum tiling problem above efficiently,
because it is directly related to BMD.
Geerts et al.'s main interest is in designing an algorithm for maximum $k$-tiling.
It can be used to solve minimum tiling problem.
In contrast, we directly attack minimum tiling in a limited search space,
as explained below.

We mentioned non-negative factorization (NMF) earlier in connection with the interpretability issue.
To address this issue from a different angle,
Drineas et al.~\cite{drineas2006,drineas2008} introduced CX- and CUR-decompositions.
In the CX-decomposition
a given matrix $M$ is decomposed into two matrices $C$ and $X$ such that
the ``difference'' between $M$ and $CX$ is minimized,
with the condition that the columns of $C$ must be a subset of the columns of $M$,
namely, the {\em column use condition} is imposed.
In the CUR-decomposition, on the other hand,
a given matrix $M$ is decomposed into three matrices $C, U$ and $R$,
with the condition that the columns of $C$ (resp. rows of $R$)
must be a subset of the columns (resp. rows) of $M$.
Miettinen applies CX- and CUR-decompositions to BMD
(where all the factor matrices, as well as $M$, are Boolean)
and proposes heuristic algorithms~\cite{miettinen2008a}.

To address the interpretability issue,
we also adopt the {\em column use condition}
that the set of columns of the factor matrix $U$ form a subset of the columns of $M$
 in our decomposition $M =U\circ V$.
Arguments in support of imposing this condition in some data mining applications
can also be found in~\cite{hyvonen2008}.
Note that in CX-decomposition,
a parameter $k$ is given
and it is required to find an optimal $C$ with $k$ columns
that minimizes the ``difference'' between $M$ and $CX$.
Therefore, the algorithms in \cite{miettinen2008a} cannot be used directly for our purpose,
which is to find $C$ and $X$ with the minimum $k$ such that
$CX$ {\em exactly} equals $M$.
In any case, imposing the column use condition has a beneficial effect of reducing the search space 
for an optimal BMD.
\hide{
Fault-tolerant patterns Besson \cite{besson2006}
Miettinen et al. \cite{miettinen2008b} show that {\em discrete basis problem} (DBP)
is NP-hard and it cannot be approximated unless P=NP.
Minimum number of tiles can be approximated within $O(\log nm)$~\cite{geerts2004}.
This requires an oracle that gives the largest tile.
}

\subsection{Main contributions of this paper}
We first derive a closed-form formula for $J$ satisfying $M=M\circ J^T$,
where $J$ is  ``maximal'' in the sense that if any 0 in $J$ is changed to a
1, then this equality is violated.
We then propose two heuristic algorithms for decomposing $M\in \{0,1\}^{m\times n}$
into $U\circ V$
such that $U\in \{0,1\}^{m\times k}$ satisfies the column use condition and its column dimension is minimized.
Matrix $J$ greatly facilitates finding a set of candidate tiles.

Two important performance criteria are (i) how close is the common dimension $k$
of the generated $U$ and $V$ to the (Schein) rank of $M$,
which is the minimum possible,
and (ii) how fast $U$ and $V$ can be computed.
We demonstrate that our algorithms do rather well in these aspects
in comparison with other known
algorithms without the column use
condition~\cite{belohlavek2013,belohlavek2010,geerts2004,xiang2011},
despite the fact that some of them are based on fairly ``sophisticated'' concepts.
Obviously, without the column use condition,
one should be able to achieve a smaller (not larger to be exact) $k$.
When the objective is {\em exact} BMD,
in spite of this restriction,
our algorithms do as well as or better than the others on four out of the five popular datasets
we have tested,\footnote{See
Table~I in Section~\ref{sec:performance}.
The rows labeled 100\% shows the data for exact decomposition.}
which we find somewhat surprising.

Our algorithms can also be used for ``from-below'' approximation\footnote{For any 1 element
in $U\circ V$, the corresponding element in $M$ must be a 1.}~\cite{belohlavek2013,belohlavek2010}
as well with competitive {\em coverage}
(i.e., the fraction of the 1's covered by the selected tiles).
Since matrix operations are available in popular mathematical software packages
such as {\tt Matlab, Maple}, and the {\tt R}-language,
we made special efforts to state our algorithms in matrix operations.
We believe that it has helped to enhance readability.

\subsection{Related work}
\hide{
}

It is clear that BMD is easily reducible to the {\em set cover} ({\em SC})  problem.
Feige~\cite{feige1998} shows that SC can be solved approximately
with the guaranteed approximation ratio of $O(\log n)$ in the worst case.
Umetani et al.~\cite{umetani2007} give a survey on SC algorithms,
but new heuristics are still being proposed, e.g., \cite{lan2007}.
Belohlavek et al.~\cite{belohlavek2010} comment that using a SC heuristic
(without any modification)
to solve BMD is not very effective.
In another context,
Miettinen also states that in practice algorithms without provable approximation
factors performed better~\cite{miettinen2008a}.

We now review the known heuristic algorithms for BMD,
which are closely related to our work reported in this paper.
Geerts et al.~\cite{geerts2004} concentrate on  `maximum $k$-tiling'
and `large tile mining' mentioned before.
Their algorithm, which we call {\tt Tiling}, 
uses the well known greedy SC heuristic to iteratively find tiles
that cover the most uncovered 1's in the given matrix $M$.
Unfortunately, it cannot be used for exact BMD.

Miettinen et al. designed an algorithm, named  {\tt Asso},
to solve the {\em discrete basis problem}~\cite{miettinen2008b}.
As such it does not find tiling, and does not exclude tiles
which may cover 0's in the matrix.
Therefore, in general, it is not suitable for finding exact BMD,
which is the main topic of this paper.
\hide{
We feel that the weight threshold hinders the optimization,
since the minimizing the number of tiles, each 1 in $M$ should have equal weight.

From \cite{belohlavek2013}:
 {\tt Asso} first computes an $m\times m$ Boolean matrix $A$ in which $A_{ij} = 1$
if the confidence of rule $\{i\} \Rightarrow \{j\}$ exceeds a parameter $\tau$.
The rows of $A$ are then used as candidates for the rows of the factor-attribute (i.e., basis vector) matrix.
The actual rows are selected using a greedy approach using function {\tt cover} that rewards
with weight $w^+$ the decrease of error $E_u$ and penalizes with weight $w^-$
 the increase of $E_o$ that is due to a given row of $A$. 
{\tt Asso} commits both types of errors, $E_u$ and $E_o$.
It thus provides general factorizations which need not be from-below factorizations of $I$.
{\tt Asso} is not meant for producing exact BMD with the minimum dimension.
}

Work by Belohlavek et al.~\cite{belohlavek2013,belohlavek2010} addresses exact
as well as approximate BMD.
They make use of  lattice theoretic concepts and ideas from {\em formal concept analysis},
and propose two heuristic algorithms, named {\tt GreConD} and {\tt GreEss},
which find good ``from-below'' approximation as well as exact BMD.
They do not impose the column use condition.
In \cite{belohlavek2013},
they compare the performance of their algorithms with other known algorithms.

Another group of researchers, Xiang et al., worked on the ``summarization'' of a database~\cite{xiang2011}.
Essentially, they also try to find a tiling that covers all 1's in a given transactional database,
which can be represented by a Boolean matrix.
However, the objective function that they want to minimize is not the number of tiles in the tiling,
but the total size of the ``description length,''
where the ``description length'' of a tile is defined as the sum of the number of 1's
in a row of the tile and the number of 1's in a column of the tile.
They propose a heuristic algorithm, named {\tt Hyper}, to minimize this objective function,
and claim that it also tends to minimize the number of tiles,
which is the dimension $k$ in our model.

\subsection{Paper organization}
The rest of the paper is organized as follows.
Section~\ref{sec:prelim} gives some basic definitions which will be used throughout the paper,
and reviews a minimal set of Boolean algebra facts needed to understand this paper.
Section~\ref{sec:BMD} is devoted to the proofs of our major mathematical results,
which form the theoretical basis for the algorithms proposed in Section~\ref{sec:algorithms}.
We propose two new algorithms for decomposing a given $M$ into the unknown $U$ and $V$,
and illustrate them with a simple example.
Section~\ref{sec:performance} presents some experimental results,
which are very encouraging.
In Section~\ref{sec:application}, as an example of possible practical applications,
we show how to apply our algorithms to educational data mining.
Section~\ref{sec:conclusion}, concludes the paper with some discussions.

\section{Preliminaries}\label{sec:prelim}
In this section the basic notations and definitions used throughout this paper are given.
We also cite some basic formulae of Boolean matrix theory.
Some standard terms in matrix theory are used without definition
since they are readily available, for example,
in books by Golub and Van Loan~\cite{golub1996} and Kim~\cite{kim1982}.

\subsection{Notations and Definitions}
Let $M= [\mu_{ij}] \in\{0,1\}^{m\times n}$.
Although there is no intrinsic size or magnitude attribute in the value
0 ({\tt False}) and 1 ({\tt True}),
we assume that the ``larger than'' ($>$) relation $1> 0$ holds
and $1-0=1, 1-1=0-0=0$.
In an expanded form, it is represented as
\begin{eqnarray}\label{eqn:matrixM}
M=\left(
\begin{array}{l}
    \bm{\mu}_1 \\
    \bm{\mu}_2 \\
    . \\
    . \\
    . \\
    \bm{\mu}_m \\
\end{array}
\right)=\left(
\begin{array}{llllll}
    \mu_{11} & \mu_{12} &... &\mu_{1n}\\
    \mu_{21} & \mu_{22} &... &\mu_{2n}\\
    . & . &. &.\\
    . & . &\ .\  &.\\
    . & . &\ \ . &.\\
    \mu_{m1} & \mu_{m2} &... &\mu_{mn}\\
\end{array}
\right)
\end{eqnarray}
where 
$
\bm{\mu}_i=[\mu_{i1},\mu_{i2},...,\mu_{in}]
$
is called the $i^{th}$ {\em row vector},
and
$
[\mu_{1j},\mu_{2j},...,\mu_{mj}]^T
$
is called the $j^{th}$ {\em column vector} of $M$.
We also often use $M[i,:]$ (resp. $M[:,j]$) to denote the $i$-th row
 (resp. $j$-th column) vector of $M$.
The matrix whose $(i,j)$ elements is  $\overline{\mu}_{ij}$,
where $\overline{0}=1$ and $\overline{1}=0$,
 is called the {\em complement} of $M$
and is denoted by $\overline{M}$.
The matrix whose $(i,j)$ elements is $\mu_{ji}$ is called the {\em transpose} of $M$,
and is denoted by $M^T$.
The $n\times n$ identity matrix is denoted by $I_{n\times n}$,
and $[0]_{m\times n}$ shall denote an ${m\times n}$ matrix whose elements are all 0's.
Let $\mathbb{R}$ (resp. $\mathbb{N}$) denote the set of all real numbers 
(resp. natural numbers, including 0). 

\begin{definition} 
Let $\bm{p}= [p_1, p_2,\ldots, p_n] \in\{0,1\}^{1\times n}$ and $\bm{q}= [q_1, q_2,\ldots,q_n] \in\{0,1\}^{1\times n}$.
We say the $\bm{p}$ {\em dominates} $\bm{q}$ if  $p_i\geq q_i$ for all $i=1,\ldots, n$,
and write $\bm{p} \geq \bm{q}$.
We write $\bm{p} > \bm{q}$ if $\bm{p} \geq \bm{q}$ and $p_{i} > q_{i}$ for some $i~(1\leq i\leq n)$,
and say that $\bm{p}$ {\em strictly dominates} $\bm{q}$.
Dominance relation is similarly defined for a pair of column vectors.
\end{definition}

\begin{definition} 
We define a partial order ``$\leq$'' on a pair of binary matrices
$P= [p_{ij}] \in\{0,1\}^{m\times n}$ and $Q= [q_{ij}] \in\{0,1\}^{m\times n}$.
We write $P\leq Q$, if $p_{ij}\leq q_{ij}$, for all $ i=1,2,\ldots, m$ and $j=1,2,\ldots,n$.
\end{definition}

\begin{definition} 
Let $P= [p_{ij}] \in\{0,1\}^{m\times n}$ and $Q= [q_{ij}] \in\{0,1\}^{m\times n}$
such that $P\leq Q$. 
We say the $P$ {\em covers} the set of 1 entries in $Q$ at $\{(i,j) \mid p_{ij} =1\}$.
\end{definition}

\begin{definition} 
If $U= [u_{ij}] \in\{0,1\}^{m\times n}$ and $V= [v_{ij}] \in\{0,1\}^{m\times n}$,
the {\em element-wise Boolean sum} of $U$ and $V$ is defined by
\[
U\vee V=[u_{ij}\vee v_{ij}] \in\{0,1\}^{m\times n},
\]
and {\em element-wise Boolean product} of $U$ and $V$ is defined by
\[
U\wedge V=[u_{ij}\wedge v_{ij}] \in\{0,1\}^{m\times n},
\]
where
$0\vee 0 = 0$, $1\vee 0 = 0\vee 1=1\vee 1 = 1$,
$0\wedge 0 =  1\wedge 0 = 0\wedge 1=0$, and $1\wedge 1 = 1$.
\end{definition}

For $U= [u_{ij}] \in\{0,1\}^{m\times k}$ and  $V= [v_{ij}] \in\{0,1\}^{k\times n}$,
 their ordinary {\em arithmetic product} is defined by
\begin{equation}
P= UV =[p_{ij}]\in \mathbb{R}^{m\times n}, ~p_{ij}=\sum_{t=1}^k u_{it}v_{tj}. \label{eqn:G}
\end{equation}
Their {\em Boolean product} is defined by
\begin{equation}
B= U\circ V = [b_{ij}]\in\{0,1\}^{m\times n},  ~b_{ij}=\vee_{t=1}^k (u_{it} \wedge v_{tj}). \label{eqn:H}
\end{equation}

In a Boolean product, 1's and 0's are considered as Boolean values,
while in an arithmetic product, they are treated as integers.
Let $M$ be given by (\ref{eqn:matrixM}) and $c$ be a constant.
The matrix whose $(i,j)$ element is $c\mu_{ij}$ is called a {\em scaler multiple} of $M$
and is denoted by $c \cdot M$. 

\subsection{Brief review of matrix algebra relevant to this paper}
The materials in this subsection, 
except Lemma~\ref{lem:prod2}, can be found in {\rm \cite{golub1996,kim1982}}.
\begin{proposition} Associativity.
\begin{enumerate}
\item[(a)]
$(UV)W=U(VW)$
\item[(b)]
$(U\circ V)\circ W =U\circ (V\circ W)$.
\end{enumerate}
\end{proposition}
We can thus write $UVW$ (resp. $U\circ V\circ W$) for (a) (resp. (b))
without ambiguity.
\begin{proposition}\label{prop:transpose} Transpose of product.
\begin{enumerate}
\item[(a)]
For $U \in\{0,1\}^{m\times k}$ and  $V\in\{0,1\}^{k\times n}$,
$(U\circ V)^T = V^T \circ  U^T$ holds.
\item[(b)]
For $U\in \mathbb{R}^{m\times k}$ and $V\in \mathbb{R}^{k\times n}$,
$(U V)^T = V^T   U^T$ holds.
\end{enumerate}\qed
\end{proposition}

\begin{proposition} \label{prop:expand} Product expansion.
\begin{eqnarray}\label{eqn:expand}
M &=&U\circ V =U[:,1]\circ V[1,:]\vee U[:, 2]\circ V[2,:]\vee \ldots \nonumber\\
 &&\vee U[:,k]\circ V[k,:] \nonumber\\
\ &=&\vee_{t=1}^k \{U[:, t]\circ V[t,: ]\} 
\end{eqnarray}
\end{proposition}
The following proposition follows directly from (\ref{eqn:H}).
\begin{proposition}\label{prop:prod1} 
Let $\bm{p}= \bm{\left[}p_1~p_2 \ldots p_m\bm{\right]}$ and
$\bm{q}= $ $\bm{\left[}q_1~q_2\ldots q_n\bm{\right]}$
be two Boolean row vectors.
We have
\hide{
}

\begin{eqnarray}\label{eqn:colrowproduct}
\bm{p}^T\circ \bm{q} &=& \bm{\left[}q_1\cdot\bm{p}^T~~q_2\cdot\bm{p}^T\ldots q_n\cdot\bm{p}^T\bm{\right]}\\
&=&\left(\begin{array}{l}
    p_1 \cdot\bm{q}\\
    p_2 \cdot\bm{q}\\
    ~~ . \\
    ~~. \\
    ~~. \\
    p_m \cdot\bm{q}\\
\end{array}
\right)  \in \{0,1\}^{m\times n}.
\end{eqnarray}
\end{proposition}

For example, if $\bm{p}= \bm{\left[}0~1~0~1~0~1\bm{\right]}$
and $\bm{q}= \bm{\left[}0~1~0~1~1\bm{\right]}$,
then
\begin{eqnarray} \label{eqn:tile}
\bm{p}^T\circ \bm{q} =\left(
\begin{array}{llllll|}
    . & . &. &. &.\\
    . &1 & . &1&1\\
    . & . &. &. &.\\
    . &1 & . &1&1\\
    . & . &. &. &.\\
    . &1 & . &1&1\\
\end{array}
\right)
\end{eqnarray}
Thus $\bm{p}^T\circ \bm{q}$ represents a tile.
We identify $\bm{p}^T\circ \bm{q}$ with the tile it represents, and sometimes call
this expression itself a tile.
The formula in the following lemma will be used to simplify our algorithms later.
\begin{lemma}\label{lem:prod2}
Let $\bm{p}= \bm{\left[}p_1~p_2 \ldots p_m\bm{\right]}$ and
$\bm{q}= $ $\bm{\left[}q_1~q_2\ldots q_n\bm{\right]}$
be two Boolean row vectors,
and let $C\in \{0,1\}^{n\times m}$.
Then the following equality holds.
\begin{equation}\label{eqn:covered}
||C \wedge (\bm{p}^T\circ \bm{q})|| = \bm{q} C \bm{p}^T.
\end{equation}
\end{lemma}
\begin{proof}
The quantity $||C \wedge (\bm{p}^T\circ \bm{q})||$ is clearly the number of 1 elements of $C$ such that
the corresponding element of $\bm{p}^T\circ \bm{q}$ is also a 1.
By Proposition~\ref{prop:prod1}, the $(i,j)$ element of $\bm{p}^T\circ \bm{q}$ is a 1
if $p_i=q_j=1$, and a 0 otherwise.
Note that $C \bm{p}^T\in \mathbb{N}^{n\times 1}$ on the right hand side of (\ref{eqn:covered})
is a column vector such that its $i^{th}$ element is the number of 1's in the $i^{th}$ row of $C$,
which are counted if it is in column $j$ satisfying $C[i,j]=p_j=1$.
Now $\bm{q}(C \bm{p}^T)$ ``picks up'' the $i^{th}$ element of $C \bm{p}^T$ provided $q_i=1$
and adds the picked up numbers.
\qed
\end{proof}

\section{BMD Theorems}\label{sec:BMD}

In the rest of this paper, we refer to matrix $U\in \{0,1\}^{m\times k}$ defined by
\begin{eqnarray}
U=\left(
\begin{array}{l}
    \bm{u}_1 \\
    \bm{u}_2 \\
    . \\
    . \\
    . \\
    \bm{u}_m \\
\end{array}
\right)=\left(
\begin{array}{llllll}
    u_{11} & u_{12} &... &u_{1k}\\
    u_{21} & u_{22} &... &u_{2k}\\
    . & . &. &.\\
    . & . &\ .\  &.\\
    . & . &\ \ . &.\\
    u_{m1} & u_{m2} &... &u_{mk}\\
\end{array} \label{eqn:matrixU}
\right)
\end{eqnarray}
and matrix $V\in \{0,1\}^{k\times n}$ defined by
\begin{eqnarray}
V=\left(
\begin{array}{l}
    \bm{v}_1 \\
    \bm{v}_2 \\
    . \\
    . \\
    . \\
    \bm{v}_k \\
\end{array}  \label{eqn:matrixV}
\right)=\left(
\begin{array}{llllll}
    v_{11} & v_{12} &... &v_{1n}\\
    v_{21} & v_{22} &... &v_{2n}\\
    . & . &. &.\\
    . & . &\ .\  &.\\
    . & . &\ \ . &.\\
    v_{k1} & v_{k2} &... &v_{kn}\\
\end{array}
\right)
\end{eqnarray}

The following lemma follows easily from the fact that $1\vee 1 = 1$.
\begin{lemma} 
Define matrices $G= [g_{ij}]=UV \in \mathbb{N}^{m\times n}$ and
$H=[h_{ij}]=U\circ V^T\in \{0,1\}^{m\times n}$.
Then for $i=1,2,\ldots,m$ and $j=1,2,\ldots,n$  we have
\begin{eqnarray}
g_{ij}=0 &\Leftrightarrow& h_{ij}=0 \nonumber\\
g_{ij} \geq 1&\Leftrightarrow& h_{ij}=1.
\end{eqnarray}
\end{lemma}

The following proposition follows easily from definition.
\begin{proposition}\label{prop:dominates}
Let $\bm{p}, \bm{q} \in\{0,1\}^{1\times a}$ be two Boolean row vectors.
Then ``$\bm{p}$ dominates $\bm{q}$'' can be expressed as
\begin{eqnarray} \label{lem:domination}
\bm{p} \geq \bm{q} \Leftrightarrow  \overline{\bm{p}}\circ\bm{q}^T = \bm{q} \circ\overline{\bm{p}}^T=0
 \Leftrightarrow \overline{\overline{\bm{p}} \circ\bm{q}^T}= \overline{\bm{q} \circ\overline{\bm{p}}^T}= 1.  \label{eqn:inclusion}
\end{eqnarray}
\end{proposition}

Lemma~\ref{lem:Jmatrix} below plays an important role in what follows.
In order to prove it, we first need to show a technical lemma. 
\begin{lemma} \label{lem:doubleComplement}
Let $P\in \{0,1\}^{a\times p}$ be an arbitrary Boolean matrix.
\begin{enumerate}
\item[(a)]
For any two row vectors $\bm{u}, \bm{v}\in \{0,1\}^{1\times a}$
we have
\begin{eqnarray}
[\overline{\bm{v}}=\overline{(\bm{u}\circ P)}\circ P^T]
\Rightarrow \bm{v}\geq \bm{u} \label{eqn:vgreater1} 
\end{eqnarray}

\item[(b)]
For any two matrices $U,V\in \{0,1\}^{b\times a}$
we have
\begin{eqnarray}
[\overline{V}= \overline{U\circ P}\circ P^T] \Rightarrow V\geq U.\label{eqn:Vgreater1}
\end{eqnarray}
\end{enumerate}
\end{lemma}

\begin{proof}
(a) Suppose $\overline{\bm{v}}=\overline{(\bm{u}\circ P)}\circ P^T$ holds.
Then $\overline{v}_j=0$ (i.e., $v_j=1$) if and only if
\[
\overline{(\bm{u}\circ P)}\circ P[j,:]^T=0. 
\]
By Proposition~\ref{prop:dominates},
this implies that $\bm{u}\circ P$ dominates the $j^{th}$ column of $P^T$, 
i.e., the $j^{th}$ row of $P$.
Since this clearly happens if $u_j=1$, we have $u_j=1 \Rightarrow v_j=1$.
It follows that $\bm{v}\geq \bm{u}$.

(b) Let $\bm{u}_i$ (resp. $\bm{v}_i$) be the $i^{th}$ row vector of matrix $U$ (resp. $V$),
as in (\ref{eqn:matrixU}) (res. (\ref{eqn:matrixV})).
Then (\ref{eqn:vgreater1}) holds for each $i~(1\leq i \leq b)$, namely,
\begin{equation}
[\overline{\bm{v}}_i=\overline{(\bm{u}_i\circ P)}\circ P^T]
\Rightarrow \bm{v}_i\geq \bm{u}_i,\nonumber 
\end{equation}
and (\ref{eqn:Vgreater1}) follows.
\qed
\end{proof}

Without loss of generality, we assume from now on that the given matrix $M$ has
no all-0 row or all-0 column.
We now prove the following theorem,
which provides a basis for the algorithms given in the next section.
\begin{lemma}\label{lem:Jmatrix}
Let $M\in\{0,1\}^{m\times n}$, $U\in\{0,1\}^{m\times k}$, and $V\in\{0,1\}^{n\times k}$
satisfy $M=U\circ V^T$,
and define
\begin{eqnarray}
J&\equiv& \overline{\overline{M}^T \circ U} \label{eqn:J1}
\end{eqnarray}
Then we have
\begin{enumerate}
\item[(a)]  $V\leq J$, and
\item[(b)] 
$M=U\circ J^T$
\end{enumerate}
\end{lemma}

\begin{proof}
 (a) From (\ref{eqn:J1}), we get
\begin{eqnarray}
\overline{J}&=& \overline{M}^T \circ U \label{eqn:J1a}
\end{eqnarray}
Plugging $M=U\circ V^T$ into (\ref{eqn:J1a})
and using Proposition~\ref{prop:transpose}(a),
we obtain 
\begin{eqnarray}
&&\overline{J}= \overline{U\circ V^T}^T \circ U  =\overline{V\circ U^T} \circ U.\label{eqn:J1b}
\end{eqnarray}
Eq. (\ref{eqn:Vgreater1}) is the same as (\ref{eqn:J1b}) if we set $P=U^T$, $V=J$, and $U=V$,
which yields $J\geq V$.

(b) Define $N=U\circ J^T$. 
We want to show that $N=M$.
From (\ref{eqn:J1}), we get
\begin{eqnarray}
&&N^T=J\circ U^T=\overline{\overline{M}^T \circ U} \circ U^T, 
\end{eqnarray}
which yields
$\overline{N}^T\geq \overline{M}^T$ or  $N\leq M$ 
by (\ref{eqn:Vgreater1}).
On the other hand, from $J\geq V$ (proved in (a) above) we get
$M=U\circ V^T\leq U\circ J^T=N$.
It follows that $M= N$.
\qed
\end{proof}

From now on, we consider the special case in Lemma~\ref{lem:Jmatrix},
where $U=M$,
hence
\begin{eqnarray}\label{eqn:J2}
J=\overline{\overline{M}^T\circ M}\in \{0,1\}^{n\times n}.
\end{eqnarray}
Lemma~\ref{lem:Jmatrix} has an important implication,
which we state as a theorem.
\begin{theorem}\label{thm:Wmaximal}
Given an arbitrary matrix $M\in\{0,1\}^{m\times n}$,
let $J$ be defined by (\ref{eqn:J2}).
Then $V\leq J$ holds for any matrix $V\in \{0,1\}^{n\times n}$ 
satisfying $M=M\circ V^T$.
\qed
\end{theorem}

Matrix $J$ has a number of other important properties.

\begin{lemma}\label{lem:JM}
For any $M\in \{0,1\}^{m\times n}$, 
matrix $J$ defined by (\ref{eqn:J2}) has the following properties.
\begin{itemize}
\item[(a)]
$J[i,j]=1\Leftrightarrow M[:,i]\geq M[:,j]$, i.e.,
column $i$ dominates column $j$ of $M$.
\item[(b)]
$J[i,j]=J[j,i]=1\Leftrightarrow M[:,i]=M[:,j]$
$\Leftrightarrow J[:,i]= J[:,j]$ and $J[i,: ]= J[j,: ]$.
\item[(c)]
$J[i,j]=1 > J[j,i]=0\Leftrightarrow M[:,i]> M[:,j]$ $\Rightarrow J[:,j]> J[:,i]$ and $J[j,:] < J[i,:]$.
\end{itemize}
\end{lemma}
\begin{proof}
(a) If we let $\bm{p}= M^T[:,i]$ and $\bm{q}= M^T[:,j]$ in 
(\ref{lem:domination}),
then we get $M^T[:,i] \geq M^T[:,j]$ if and only if
\[
\overline{\overline{M}^T[:,i]\circ M[:,j]}=1,
\]
which holds if and only if $J[i,j]=1$ by (\ref{eqn:J2}).

(b) By interchanging $i$ and $j$ in part (a),
we get $J[j,i]=1\Leftrightarrow M[:,i]\leq M[:,j]$.
It follows that 
$J[i,j]=J[j,i]=1\Leftrightarrow M[:,i]=M[:,j]$.
Thus any column that dominates $M[:,j]$ also dominates $M[:,i]$,
and vice versa,
namely $J[:,i]= J[:,j]$.
Moreover,
 any column that is dominated by $M[:,j]$ is also dominated by $M[:,i]$,
and vice versa,
namely $J[i,:]= J[j,:]$.

(c) $J[i,j]=1 > J[j,i]=0$ implies that $M[:,i]$ strictly dominates $M[:,j]$, i.e., $M[:,i]> M[:,j]$.
In this case, any column of $M$ that dominates $M[:,i]$ also dominates $M[:,j]$,
which implies $J[:,j]> J[:,i]$,
and any column of $M$ that is dominated by $M[:,j]$ is also dominated by $M[:,i]$,
hence $J[j,:] < J[i,:]$.
\qed
\end{proof}

The properties proved in Lemma~\ref{lem:JM} can be verified in the following example.
\begin{example}\label{ex:ex1}
\[
M=\left(
\begin{array}{lllll}
    1 & 1 & 1 & 1 &1 \\
    0 & 0 & 1 & 1 &0 \\
    1 & 1 & 0 & 0 &1 \\
    1 & 0 & 0 & 1 &1 \\
\end{array}
\right);~~~
J=\overline{\overline M^T \circ M}=\left(
\begin{array}{lllll}
    1 & 1 & 0 & 0 &1 \\
    0 & 1 & 0 & 0 &0 \\
    0 & 0 & 1 & 0 &0 \\
    0 & 0 & 1 & 1 &0 \\
    1 & 1 & 0 & 0 &1 \\
\end{array}
\right)
\]
\smartqed
\end{example}

We now prove another useful property of matrix $J$.
\begin{lemma}\label{lem:J}
Given an arbitrary matrix $M\in\{0,1\}^{m\times n}$,
let $J$ be defined by (\ref{eqn:J2}).
If any 0-element in $J$ is changed to a 1,
then $M=M\circ J^T$ no longer holds.
\end{lemma}
\begin{proof}
Assume that $J$ does not have the maximum number of 1's and assume that
$J[i,j]=0$, $1\leq i,j \leq n$, can be changed from 0 to 1 without violating
Lemma~\ref{lem:Jmatrix}(b) with $U=M$, i.e., $M=M\circ J^T$.
Let $\bm{w}_j = J[j,\cdot]$, so that $(\bm{w}_j)^T$ is the $j^{th}$ column of $J^T$.
If the $i^{th}$ element of $\bm{w}_j$ is $0$, i.e., $J[j,i]=0$,
then $M[\cdot,j] \not\geq M[\cdot,i]$ by Lemma~\ref{lem:JM}(a). 
Let $\bm{w'}_j$ be obtained from $\bm{w}_j$ by changing its $i^{th}$ element
from 0 to 1.
Since $M\circ (\bm{w}'_j)^T \geq M[\cdot,i]$,
we have $M[\cdot,j]\not\geq M\circ (\bm{w}'_j)^T$,
\hide{
Let $\bm{w}_j$ be the $j^{th}$ column vector of $J$,
i.e., $\bm{w}_j = J[\cdot,j]$.
The $i^{th}$ element of $\bm{w}_j$ (i.e.,
$J[i,j])=0$ implies that $M[\cdot,i] \not\geq M[\cdot,j]$ by Lemma~\ref{lem:JM}. 
Let $\bm{w'}_j$ be obtained from $\bm{w}_j$ by changing its $i^{th}$ element
from 0 to 1.
Since $M\circ \bm{w}'_j \geq M[\cdot,i]$,
we have $M\circ \bm{w}'_j \not\leq M[\cdot,j]$,
}
a contradiction.
\hide{
Let $\bm{w}_i$ be the $i^{th}$ row vector of $J$,
i.e.,
$\bm{w}_i = J[i,\cdot]$.
The $j^{th}$ element of $\bm{w}_i$, i.e.,
$J[i,j]=0$ implies that $M[\cdot,i] \not\leq M[\cdot,j]$ by Lemma~\ref{lem:JM}. 
From (\ref{eqn:expand}) we have 
\begin{equation}\label{eqn:M}
M=\vee_{t=1}^n \{M[\cdot, t]\circ J[t,\cdot ]\}
\end{equation}
Let $\bm{w'}_i$ be obtained from $\bm{w}_i$ by changing its $j^{th}$ element
from 0 to 1.
The term corresponding to $t=i$ in (\ref{eqn:M}) is
$M[\cdot, i]\circ J[i,\cdot]$.
The $j^{th}$ column of $M[\cdot, i]\circ \bm{w}_i$ is an all-0 column vector,
while $j^{th}$ column of $M[\cdot, i]\circ \bm{w'}_i$ equals $M[\cdot, i]$.
(See (\ref{eqn:colrowproduct}).)
Since we have $M[\cdot,i] \not\leq M[\cdot,j]$ by our assumption,
$M[\cdot, i]\circ \bm{w'}_i$ has a 1 in the entry, where $M$ has a 0,
a contradiction.
}
We conclude that if any element in $J$ is changed from a 0 to a 1,
then $M=M\circ J^T$ is violated.
\qed
\end{proof}

\begin{theorem}\label{thm:optimality}
Let $M=U\circ V$ be an optimal decomposition of $M$, satisfying the column use condition,\footnote{By definition,
this means that the columns of $U$ are some of the columns of $M$.}
where $U\in\{0,1\}^{m\times k}$, $V\in\{0,1\}^{k\times n}$
and $k$ is the minimum possible.
Then for each $i=1,2,.\ldots, k$,
we have $U[:,i]\circ V[i,:] \in \{M[:,t]\circ J[t,:] \mid t=1,\ldots, n\}$.
\end{theorem}
\begin{proof}
Let 
\[
U\circ V =\vee_{i=1}^k \{U[:,i]\circ V[i,:]\},
\]
and consider a particular term $U[:,i]\circ V[i,:]$ in it.
Since $U$ consists of columns of $M$,
there is an $h$ such that $U[:,i] =M[:,h]$.
By Theorem~\ref{thm:Wmaximal},
$J[h,:]$ is the maximal row vector such that
$U[:,i] \circ J[h,:] \leq M$,
hence $V[i,:] \leq J[h,:]$.
We thus have $U[:,i]\circ V[i,:] \leq M[:,h] \circ J[h,:]$.
\qed
\end{proof}

Intuitively, Theorem~\ref{thm:optimality} implies that the search space for an optimal
decomposition of $M$ under the column use condition can be limited to
$\{M[:,t]\circ J[t,:] \mid t=1,\ldots, n\}$.
From this theorem, it is apparent that BMD is also closely related to the set covering problem.
In the next section, 
we design heuristic algorithms for exact BMD,
based on Theorem~\ref{thm:optimality}.

\section{Heuristic BMD Algorithms}\label{sec:algorithms}
\subsection{Algorithm description}\label{sec:describe}
In this section, we propose new algorithms for finding factor matrices
$U\in\{0,1\}^{m\times k}$ and $V\in\{0,1\}^{k\times n}$ from matrix $M\in\{0,1\}^{m\times n}$.
By Theorem~\ref{thm:optimality},
we want to find a subset of $\{M[:,t]\circ J[t,:] \mid t=1,\ldots, n\}$ that provides the
optimal tiling.
Since an exhaustive search is obviously impractical,
we want to devise a heuristic algorithm that yields a good suboptimal tiling. 

Suppose there exists an $l$ satisfying
\begin{equation}
U\circ V =\vee_{i=1,j\neq l}^k \{U[:,i]\circ V[i,:]\},\nonumber
\end{equation}
in other words,
\begin{equation}\label{eqn:redundant2}
\vee_{i=1,j\neq l}^k \{U[:,i]\circ V[i,:]\} \geq U[:,l]\circ V[l,:].
\end{equation}

Then we can safely eliminate the $l^{th}$ column $U[:,l]$ and the $l^{th}$ row $V[l,:]$
from $U$ and $V$, respectively,
which helps reduce the dimension $k$.
The condition (\ref{eqn:redundant2}) is equivalent to $||{\cal T}|| = ||{\cal T} - T_l||$,
where ${\cal T}=U J^T$ (arithmetic matrix product defined by (\ref{eqn:G})) with $J$ given
in (\ref{eqn:J2}),
and $T_l = U[:,l]\circ V[:,l]$.
There may be several indices $l$ that satisfy (\ref{eqn:redundant2}).
Therefore, we need to decide in which order the eliminations should be carried out.
We thus define the {\em selection index} $\sigma_i$ as follows:
\[
\sigma_i=||U[:,i]||\times ||V[:,i]||,
\]
where, as the reader recalls,
 $||V||$ represents the the number of 1's ($l_0$ norm) in vector $V$.
Clearly, $\sigma_i$ is the number of 1 entries in $M$ that are covered by $G_i$.
Given the initial matrices $U$ and $V$,
satisfying $M=U\circ V^T$,
we generate the new matrix $J$ by (\ref{eqn:J2}).
There are at least two approaches that appear reasonable,
regarding which attribute we should process first.
\begin{enumerate}
\item[(a)] {\tt Remove-Smallest}:
Remove attribute $j$ such that $\sigma_j$ is the smallest,
provided the removal does not affect $M$.
\item[(b)] {\tt Pick-Largest}:
Retain attribute $j$ such that $\sigma_j$ is the largest.
\end{enumerate}

Our first algorithm adopts strategy (a).
After deleting one attribute, we update $U$ and $V$,
and repeat the elimination process until there is no more attribute that can be deleted.

\begin{algorithm} {\tt Remove-Smallest} \label{alg:framework1} 

{\em Input:} Response matrix $M\in \{0,1\}^{m\times n}$.

\begin{enumerate}
\item
Initialize $U=M$ and $k=n$.
\item
Compute
    \begin{equation}\label{eqn:matrixV1}
      V^T=J=\overline{\overline{M}^T\circ M}.
    \end{equation}
\item
Compute \footnote{Intuitively,
${\cal T}[i,j]$ is the number of tiles in $U \circ V$ that cover $M[i,j]$.}
\[
{\cal T}=U V.
\]
\item
For $i=1,2,\ldots, k$ compute the size of the maximal tile for $i^{th}$ attribute
($\alpha_i$)  by
\[
 \sigma_i=||U[:, i]||\times  ||V[:,i]||,
\]
and rename 
 the attributes so that $\sigma_1\leq\sigma_2\leq\ldots \leq\sigma_k$ holds.
\item
For $j=1,2,\ldots,k$, do
    \begin{enumerate}
          \item Compute 
          \[
          T_j =U[:,j] \circ V[j,:]; 
          \]
          \item If $||{\cal T}|| = ||{\cal T} - T_j||$ then
           (i) remove column $U[:,j]$
          from $U$ and row $V[j,:]$ from $V$;
           and (ii) set ${\cal T}={\cal T}-T_j$; $k=k-1$.
    \end{enumerate}
\item
 Output $U$ and $V^T$. 
 \end{enumerate}
\end{algorithm}

Our second algorithm adopts strategy (b).
\begin{algorithm} {\tt Pick-Largest} \label{alg:framework3} 

{\em Input:} Response matrix $M\in \{0,1\}^{m\times n}$.
\begin{enumerate}
\item
Initialize $U=M$ and $k=n$.
\item
Compute
 \begin{equation}\label{eqn:matrixV2}
V^T=J=\overline{\overline{M}^T\circ M}.
 \end{equation}
\item
Initialize\footnote{Matrix $C$ keeps track of the 1 elements in $M$ covered by the products
that have been picked so far.}
 $C=[0]_{m\times n}\in \{0,1\}^{m\times n}$.
\item
For $i=1,2,\ldots, k$ compute the size of the maximal  tile for the $i^{th}$ attribute
($\alpha_i$) by
\[
 \sigma_i=||U[:,i]||\times  ||V[:,i]||.
\]
\item
For each $i$ such that $\alpha_i$ has not been picked or discarded,
compute (see (\ref{eqn:covered}))
\[
 \delta_i=\sigma_i - U[:,i]^TCV[:,i].
\]
If $\delta_i =0$ then
remove $\alpha_i$ by deleting $U[:,i]$ (resp. $V[i,:]$) from $U$ (resp. $V$).
\item
Let $\delta_j =\max_i \{\delta_i\}$ and compute
\[
 T_j= U[:,j] \circ V[j,:].
\]
 Update matrix $C = C \vee T_j$, and delete $U[:,j]$ (resp. $V[j,:]$) from $U$ (resp. $V$).
If there are still attributes remaining, then 
go to Step 5.
\item
 Output $U$ and $V$.  
\qed
\end{enumerate}
\end{algorithm}

\subsection{Simple example}\label{sec:examples}
\begin{example} \label{ex:ex2}
Let us consider the following matrix $M$,
and carry out Steps~2) and 4), which are common to both algorithms.
\[
M=\left(\begin{array}{lllllll}
    0 & 0 & 0 & 0 & 0 & 0 & 0\\
    1 & 0 & 1 & 1 & 0 & 1 & 1 \\
    0 & 1 & 1 & 0 & 1 & 0 & 1 \\
    0 & 0 & 0 & 1 & 0 & 0 & 0 \\
    0 & 1 & 1 & 0 & 1 & 0 & 1\\
    0 & 1 & 1 & 0 & 1 & 0 & 1 \\
    1 & 0 & 1 & 1 & 0 & 1 & 1 \\
    1 & 1 & 1 & 1 & 1 & 1 & 1\\
     \end{array}
\right)
\]
\[
V^T=\overline{\overline{M}^T\circ M}
=\left(\begin{array}{lllllll}
    1 & 0 & 0 & 0 & 0 & 1 & 0 \\
    0 & 1 & 0 & 0 & 1 & 0 & 0 \\
    1 & 1 & 1 & 0 & 1 & 1 & 1 \\
    1 & 0 & 0 & 1 & 0 & 1 & 0 \\
    0 & 1 & 0 & 0 & 1 & 0 & 0 \\
    1 & 0 & 0 & 0 & 0 & 1 & 0 \\
    1 & 1 & 1 & 0 & 1 & 1 & 1 \\
          \end{array}
\right)
\]
\begin{table}[tbh]
\centerline{
\begin{tabular}{|c||r|r|r|r|r|r|r|r|}
\hline
$i$			&1 &2 &3 &4 &5 &6 &7\\
\hline\hline
$||U[:,i]||$    &3 &4 &6 &4 &4 &3 &6\\ 
\hline
$||V[i,:]||$    &5 &4 &2 &1 &4 &5 &2\\ 
\hline
$\sigma_i$       &15&16&12&4 &16&15&12\\
\hline
\end{tabular}
}
\medskip
\caption{Computing $\sigma_i$.}
\label{tbl:sigmaA}
\end{table}

Step 3 of {\tt Remove-Smallest} computes
\begin{equation}\label{eqn:step3}
{\cal T}= U V
=\left(\begin{array}{lllllll}
    0 & 0 & 0 & 0 & 0 & 0 & 0\\
    2 & 0 & 4 & 3 & 0 & 2 & 4 \\
    0 & 2 & 4 & 0 & 2 & 0 & 4 \\
    0 & 0 & 0 & 1 & 0 & 0 & 0 \\
    0 & 2 & 4 & 0 & 2 & 0 & 4\\
    0 & 2 & 4 & 0 & 2 & 0 & 4 \\
    2 & 0 & 4 & 3 & 0 & 2 & 4 \\
    2 & 2 & 6 & 3 & 2 & 2 & 4\\
     \end{array}
\right)
\end{equation}
If we order the columns of $U$ from the smallest to the largest
according to the value of $\sigma_i$, we get 4,3,7,1,6,2,5.
Thus,  {\tt Remove-Smallest} processes the columns of $U$
in this order.

\noindent
Step 5(a): Compute $T_4$.
\[
T_4=U[:,4] \circ{V[4,:]}
=\left(\begin{array}{lllllll}
    0 & 0 & 0 & 0 & 0 & 0 & 0\\
    0 & 0 & 0 & 1 & 0 & 0 & 0 \\
    0 & 0 & 0 & 0 & 0 & 0 & 0 \\
    0 & 0 & 0 & 1 & 0 & 0 & 0 \\
    0 & 0 & 0 & 0 & 0 & 0 & 0\\
    0 & 0 & 0 & 0 & 0 & 0 & 0 \\
    0 & 0 & 0 & 1 & 0 & 0 & 0 \\
    0 & 0 & 0 & 1 & 0 & 0 & 0\\
     \end{array}
\right)
\]

\noindent
Step 5(b): $||{\cal T}|| > ||{\cal T} - T_4|| \Rightarrow$
Cannot remove attribute 4. 
 
\noindent
Step 5(a): Now try the next smallest attribute 3,
and compute $T_3$.
\[
T_3=U[:,3] \circ{V[3,:]}
=\left(\begin{array}{lllllll}
    0 & 0 & 0 & 0 & 0 & 0 & 0\\
    0 & 0 & 1 & 0 & 0 & 0 & 1 \\
    0 & 0 & 1 & 0 & 0 & 0 & 1  \\
    0 & 0 & 0 & 0 & 0 & 0 & 0 \\
    0 & 0 & 1 & 0 & 0 & 0 & 1 \\
    0 & 0 & 1 & 0 & 0 & 0 & 1  \\
    0 & 0 & 1 & 0 & 0 & 0 & 1  \\
    0 & 0 & 1 & 0 & 0 & 0 & 1 \\
     \end{array}
\right)
\]

\noindent
Step 5(b):
$||{\cal T}|| = ||{\cal T} - T_3||  \Rightarrow$ Remove attribute 3,
and update ${\cal T}$.
\[
{\cal T}={\cal T}-T_3
=\left(\begin{array}{lllllll}
    0 & 0 & 0 & 0 & 0 & 0 & 0\\
    2 & 0 & 3 & 3 & 0 & 2 & 3 \\
    0 & 2 & 3 & 0 & 2 & 0 & 3 \\
    0 & 0 & 0 & 1 & 0 & 0 & 0 \\
    0 & 2 & 3 & 0 & 2 & 0 & 3\\
    0 & 2 & 3 & 0 & 2 & 0 & 3 \\
    2 & 0 & 3 & 3 & 0 & 2 & 3 \\
    2 & 2 & 5 & 3 & 2 & 2 & 3\\
     \end{array}
\right)
\]
Similarly, attributes (columns of $M$) 7, 1 and 5 are removed.

Step 6: generates
\begin{equation}\label{eqn:results}
U=\left(\begin{array}{lll}
    0 & 0 & 0 \\
    1& 1 & 0  \\
    0 & 0 & 1 \\
    1 & 0 & 0 \\
    0 & 0 & 1 \\
    0 & 0& 1 \\
    1 & 1 & 0 \\
    1 & 1& 1 \\
     \end{array}
\right);~~~
V
=\left(\begin{array}{lllllll}
    0 &0 &0 &1 &0 &0 &0 \\
    1 &0 &1 &1 &0 &1 &1 \\
    0 &1 &1 &0 &1 &0 &1 \\
     \end{array}
\right)
\end{equation}
The columns of $U$ are columns 4, 6, and 2 of $M$,
and $M=U \circ V$. 

Let us now apply Algorithm~{\tt Pick-Largest} to matrix $M$.
We already illustrated the first four steps above.
From Table~\ref{tbl:sigmaA} we see that $\delta_5=\sigma_5=16$ is the largest
(tied with $\delta_2=\sigma_2=16$).
Since $\delta_i=0$ holds for no $i$,
we proceed to Step 6.

\[
T_5=U[:,5] \circ{V[5,:]}
=\left(\begin{array}{lllllll}
    0 & 0 & 0 & 0 & 0 & 0 & 0\\
    0 & 0 & 0 & 0 & 0 & 0 & 0 \\
    0 & 1 & 1 & 0 & 1 & 0 & 1 \\
    0 & 0 & 0 & 1 & 0 & 0 & 0 \\
    0 & 1 & 1 & 0 & 1 & 0 & 1\\
    0 & 1 & 1 & 0 & 1 & 0 & 1 \\
    0 & 0 & 0 & 0 & 0 & 0 & 0 \\
     0 & 1 & 1 & 0 & 1 & 0 & 1\\
     \end{array}
\right)
\]
We set $C= C\vee T_5$ to remember the 1's that are now covered by the picked product term. 
Although this algorithm does not use ${\cal T}$ in (\ref{eqn:step3}),
it is instructive to interpret Steps~5 and 6 of {\tt Pick-Largest} in terms of ${\cal T}$.
We have
\[
{\cal T}= {\cal T}-T_5
=\left(\begin{array}{lllllll}
    0 & 0 & 0 & 0 & 0 & 0 & 0\\
    2 & 0 & 4 & 3 & 0 & 2 & 4 \\
    0 & 1 & 3 & 0 & 1 & 0 & 3 \\
    0 & 0 & 0 & 0 & 0 & 0 & 0 \\
    0 & 1 & 3 & 0 & 1 & 0 & 3\\
    0 & 1 & 3 & 0 & 1 & 0 & 3 \\
    2 & 0 & 4 & 3 & 0 & 2 & 4 \\
    2 & 1 & 5 & 3 & 1 & 2 & 3\\
     \end{array}
\right)
\]
In Step 5, we update $\{\delta_i\}$.
For example, let us compute $CV[i,:]^T$ for $i=2$.
We get
\[
CV[2,:]^T = \left[0~0~4~0~4~4~0~4\right] \mbox{~and~} U[:,2]^TCV[:,2] =16.
\]
Therefore, $\delta_2= \sigma_2 -16 =0$.
This implies that $T_2 \leq C$. 
We can simply remove attribute $2$
(i.e, $U[:,2]$ and {$V[2,:]$}). 
Updating $C$ by $C=C\vee T_2$ does not change $C$.

\[
{\cal T}= {\cal T}-T_2
=\left(\begin{array}{lllllll}
    0 & 0 & 0 & 0 & 0 & 0 & 0\\
    2 & 0 & 4 & 3 & 0 & 2 & 4 \\
    0 & 0 & 2 & 0 & 0 & 0 & 2 \\
    0 & 0 & 0 & 1 & 0 & 0 & 0 \\
    0 & 0 & 2 & 0 & 0 & 0 & 2\\
    0 & 0 & 2 & 0 & 0 & 0 & 2 \\
    2 & 0 & 4 & 3 & 0 & 2 & 4 \\
    2 & 0 & 4 & 3 & 0 & 2 & 2\\
     \end{array}
\right)
\]

This computation can be done by matrix operation,
although it is not the most efficient, since it computes elements that
are of no use to us.
Construct a column vector ${\it Us}$ whose $i$th element is $||U[:,i]||$,
and a row vector ${\it Vs}$ whose $i$th element is $||V[:,i]||$.
Compute matrix $P= {\it Us}\circ Vs$.
\[
P
=\left(\begin{array}{rrrrrrr}
15   &12    &6    &~3   &12  &15    &6\\
20   &16    &8    &~4   &16   &20    &8\\
30   &24   &12    &~6   &24   &30   &12\\
20   &16    &8    &~4   &16   &20    &8\\
20   &16    &8    &~4   &16   &20    &8\\
15   &12    &6    &~3   &12   &15    &6\\
30   &24   &12    &~6   &24   &30   &12\\
     \end{array}
\right)
\]
Thus the diagonal elements of $P$ are $||U[:,i]||\times ||V[:,i]||$,
which are listed in Table~\ref{tbl:sigmaA}.
Note that the $i^{th}$ diagonal element of $U^T \circ C \circ V^T$ is the number of
1's in $U[:,i]\circ V[:,i]$ that are already covered by $C$.
\[
U^T \circ C \circ V
=\left(\begin{array}{rrrrrrr}
2    &4    &~2    &~0    &4    &~2    &~2\\
8   &16   &~8    &~0   &16   &~8    &~8\\
8   &16   &~8    &~0   &16   &~8    &~8\\
2    &4    &~2    &~0    &4    &~2    &~2\\
8   &16   &~8    &~0   &16   &~8    &~8\\
2    &4    &~2    &~0    &4    &~2    &~2\\
8   &16   &~8    &~0   &16   &~8    &~8\\
     \end{array}
\right)
\]
Thus the amounts $\{\delta_i\}$ can be found on the diagonal of
$P-U^T \circ C \circ V$, and they are 13, 0, 4, 4, 0, 13, 4.
So $\delta_2=13$ and $\delta_6=13$
are the largest. 
Let us pick attribute 6,\footnote{
When there is a tie in the sizes $\delta_i$,
as in this example,
there are choices as to which one we remove or pick first.  
A particular choice may affect the coverage performance.
We randomly select one.
}
 update $C=C\vee T_6$,
and recompute $P-U^T \circ C \circ V$.
Since updated $\delta_1= 0$, we discard attribute 1. (Step 5.)
We then get $\delta_7=4$, so pick attribute 7.
Since $\delta_3= 0$, we discard attribute 3.
Finally, we need to pick attribute 4.
For this particular example,
{\tt Pick-Largest} generates the same decomposition as {\tt Remove-Smallest} 
given in (\ref{eqn:results}).
\qed
\end{example}

\begin{comment}
Although computing $U^T \circ C \circ V$ is a conceptually neat way of finding $\{\delta_i\}$,
the time to compute the off-main diagonal elements is wasted.
Thus, we do not use it in {\tt Pick-Largest}.
\qed
\end{comment}

\section{Performance}\label{sec:performance}
\subsection{Complexity analysis}\label{sec:complexity}
The time complexity of both algorithms is dominated by the time to compute
matrix $V$ of (\ref{eqn:matrixV1}) and (\ref{eqn:matrixV2}),
respectively, in their Step 2.
By Proposition~\ref{prop:expand}, it can be expanded into $n$ (column vector,
row vector) pairs of sizes $m$ and $n$, respectively.
Then evaluating $\overline{M}^T \circ M$ takes time proportional to
\[
\sum_{i=1}^m ||\overline{M}^T[:,i]||\times ||M[i,:]|| \leq n \sum_{i=1}^m ||M[i,:]|| = n||M||.
\]
This implies that (\ref{eqn:matrixV1}) and (\ref{eqn:matrixV2}) can be evaluated in $O(n||M||)$ time.
Note that in terms of $\cal T$ defined in Step 3 of Algorithm~{\tt Remove-Smallest},
we have 
\[
||{\cal T}||_{l_1} = \sum_{i=1}^n||U[:,i]||\times ||V[:,i]||  \leq m \sum_{i=1}^n ||U[:,i]|| = m||M||,
\]
where $||{\cal T}||_{l_1}$ ($l_1$ norm) represents the sum of the elements of ${\cal T}$.
\begin{theorem}\label{thm:smallest}
Both Algorithms~{\tt Remove-Smallest} and\\ {\tt Pick-Largest} run in $O(m||M||)$ time.
\end{theorem}
\begin{proof}
We can consider that every operation in Algorithm~{\tt Remove-Smallest}
essentially accesses/modifies some element of ${\cal T}$ and the $(i,j)$ element
is accessed ${\cal T}[i,j]$ times.
Therefore, the total time is given by $O(||{\cal T}||_{l_1})$ $= O(m||M||)$.

As for Algorithm~ {\tt Pick-Largest},
although $\cal T$ is not used in it, imagine that it was defined.
We use $U[:,i]^TCV[:,i]$ to describe Step 5,
but it is used for only for the purpose of a concise description,
and this step can be implemented more efficiently without matrix multiplication.
All we need is a way to keep track of which 1 elements of $M$ has already been covered.
Therefore, the total time is still given by $O(||{\cal T}||_{l_1}) = O(m||M||)$,
as reasoned above.
\qed
\end{proof}

The above theorems imply that our algorithms run faster if the given matrix $M$ is sparse.
If we use a sophisticated  algorithm, 
matrix multiplication can be done in $O(m^{2.373})$ time,
assuming $m\geq n$~\cite{legall2014,williams2012}. 

We should mention that another important performance measure
for heuristic algorithms of the {\em approximation ratio} relative to the optimum.
We have not looked into this performance measure yet.

\subsection{Experiments on real datasets}\label{sec:experiments}
To evaluate the performance of our heuristic algorithms,
{\tt Pick-Largest} and {\tt Remove-Smallest},
we have tested them on several real datasets,
which have been used by other authors as benchmarks.
They are {\tt Mushroom}~\cite{lichman2013}, 
{\tt DBLP}\footnote{http://www.informatik.uni-trier.de/\textasciitilde ley/db/},
{\tt DNA}~\cite{myllykangas2006},
{\tt  Chess}~\cite{lichman2013}, and
{\tt Paleo}\footnote{http://www.helsinki.fi/science/now/}.
Table~II in the next page lists the results of our experiments and compares them with
{\tt Tiling}~\cite{geerts2004},
{\tt Asso}~\cite{miettinen2008b},
{\tt Hyper}~\cite{xiang2011},
and {\tt GreConD}~\cite{belohlavek2010}, and {\tt GreEss}~\cite{belohlavek2010}.
All but the last two columns of Table~II are from~\cite{belohlavek2013}.
The common dimension $k$ of the factor matrices,
generated by the exact BMD heuristics mentioned above are listed.
The numbers in bold face indicate the best value in each row.
The rows labeled 100\% shows the data for exact decomposition.
{\tt Asso} is not meant for exact BMD, as commented earlier.

\begin{table*}[htb]\label{tbl:coverage}
\centering
\begin{tabular}{|p{18mm}|r|rrrrrrr|}
\hline
				&Coverage&{\tt Tiling}&{\tt Asso}&{\tt Hyper}&{\tt GreConD}&{\tt GreEss}&{\tt Remove-Smallest}&{\tt Pick-Largest}\\
\hline\hline
{\tt Mushroom}~\cite{lichman2013}&50\% 	&7		&\bf{6}	&19		&7		&8		&37		&10\\
\cline{2-9}
(8,124$\times$119)	&75\% 	&\bf{24}	&36		&37		&\bf{24}	&26		&59		&27\\	
\cline{2-9}
				&100\%	&119		&N/A		&122		&120		&\bf{105}	&109		&109\\
\hline
{\tt DBLP}			&50\% 	&\bf{5}	&\bf{5}	&\bf{5}	&\bf{5}	&\bf{5}	&6		&6\\
\cline{2-9}
(6,980$\times$19)	&75\% 	&\bf{10}	&\bf{10}	&\bf{10}	&11		&\bf{10}	&11		&11\\	
\cline{2-9}
				&100\%	&21		&19		&\bf{19}	&20		&\bf{19}	&\bf{19}	&\bf{19}\\
\hline
{\tt DNA}~\cite{myllykangas2006}&50\% 	&32		&\bf{27}	&67		&33		&41		&67		&58\\
\cline{2-9}
(4,590$\times$392)	&75\% 	&94		&\bf{80}	&155		&96		&105		&155		&123\\	
\cline{2-9}
				&100\%	&489		&N/A		&392		&496		&408		&\bf{368}	&\bf{368}\\
\hline
{\tt  Chess}~\cite{lichman2013}&50\% 	&5		&\bf{2}	&26		&4		&6		&26		&12\\
\cline{2-9}
(3,196$\times$75)			&75\% 	&16		&\bf{15}	&39		&15		&17		&44		&26\\	
\cline{2-9}
						&100\%	&124		&N/A		&90		&124		&113		&\bf{72}	&\bf{72}\\
			
\hline
{\tt Paleo}			&50\% 	&39		&40		&\bf{38}	&39		&\bf{38}	&39		&39\\
\cline{2-9}
(501$\times$139)	&75\% 	&75		&76		&\bf{73}	&76		&\bf{73}	&75		&74\\		
\cline{2-9}
				&100\%	&151		&N/A		&\bf{139}	&152		&145		&\bf{139}	&\bf{139}\\
\hline
\end{tabular}
\medskip
\caption{Coverage comparison of BMD algorithms for five datasets.}
\end{table*}

Among the datasets we used,
{\tt Mushroom} consists of 8,124 objects and 23 ``nominal'' attributes.
If a ``nominal'' attribute $y$ takes $k>2$ values, $\{v_1, . . . , v_k\}$,
we expanded $y$, replacing it by $k $ Boolean attributes $\{y_{v_1}, \ldots , y_{v_k}\}$
in such a way that in each row $i$ the value of the column corresponding to $y_{v_j}$ is 1
if the value of the attribute $y$ in row $i$ in the original dataset is equal to $v_j$.

Note that only our algorithms impose the column use condition.
In spite of this restriction, {\tt Pick-Largest} achieves the smallest tiling size
(or dimension $k$) for exact coverage
for  four out of the five datasets in Table II,
which was somewhat unexpected.
Incidentally, we have found a decomposition without the column use condition with $k=101$
by some other means,
so none of the algorithms in Table~II can find the optimal decomposition for the {\tt Mushroom} dataset.
As can be seen from Table~II,
{\tt Pick-Largest} and {\tt Remove-Smallest} performed equally well in finding the
 exact decomposition.

Although our original intention was to design algorithms for exact BMD,
our algorithms can also be used for ``from-below''
approximation~\cite{belohlavek2010}.
In the ``from-below'' approximation, an important performance criterion is the {\em coverage}
defined as the number of 1's covered by the product $U\circ V$ over the total number of 1's
in the given $M$~\cite{geerts2004}.
The coverage is given in the second column of Table~II.
Each entry in the table is the number of tiles used,
which is the same as the common dimension $k$ of $U$ and $V$.
Fig.~1 
plots the coverage of Algorithm {\tt Pick-Largest} as a function of the number of attributes
contained in $U$ and $V$.
The attributes are arranged in the order they were picked.
\begin{figure}[hbt]\label{fig:picklargest}
\includegraphics[width=8.5cm]{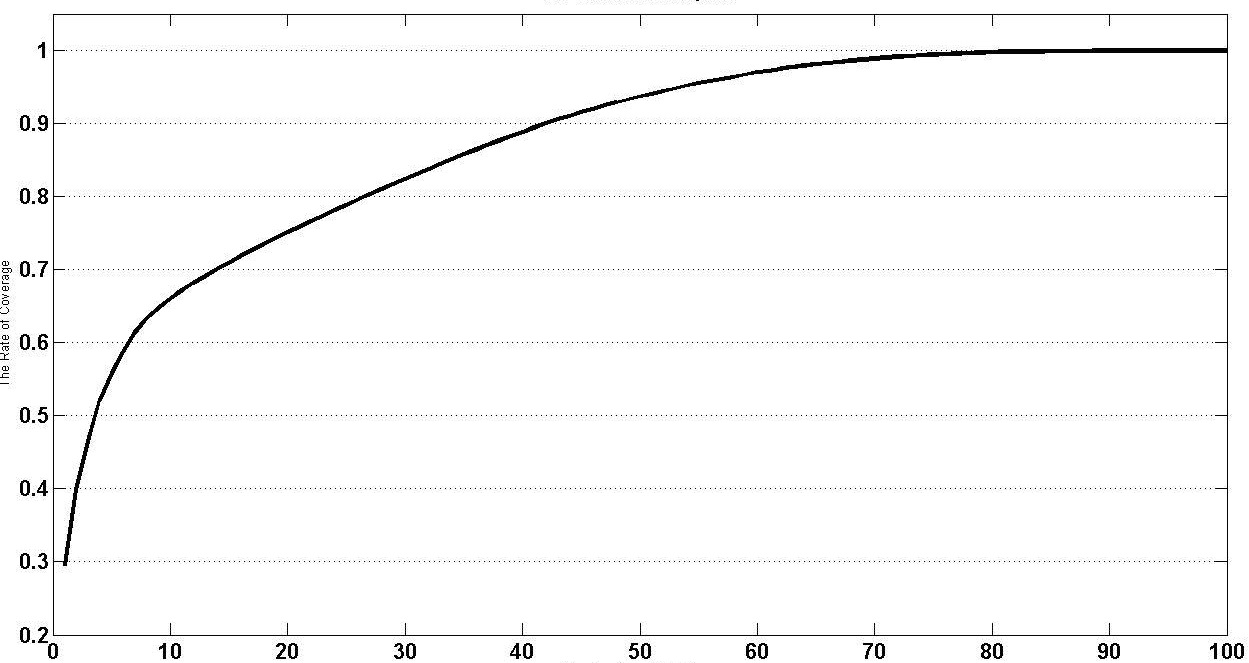}
\caption{Coverage of Algorithm {\tt Pick-Largest} for {\tt Mushroom}.}
\end{figure}

In most applications, high coverage, say, more than 90\%, would be of interest,
and we have collected coverage data in Table III 
 in this range for {\tt Pick-Largest} and {\tt Remove-Smallest},
but unfortunately not for the others,
since we haven't had the time to program the other algorithms.
We have some evidence to suggest that our algorithms perform better than others
especially at higher coverages.

\begin{table}[htb]\label{tbl:coverage1}
\centering
\begin{tabular}{|p{19mm}|c|ccp{3.5mm}c|}
\hline
				&Coverage	&{\tt Mushroom}	&{\tt DBLP}	&{\tt DNA}		&{\tt Paleo}\\
\hline\hline
				&90\% 		&76				&15			&243			&107\\
\cline{2-6}
{\tt Rem.-}{\tt Smallest}&95\% 		&85				&17			&292			&112\\	
\cline{2-6}
				&98\%		&100				&19			&332			&132	\\
\hline
				&90\% 		&47				&15			&197			&107\\
\cline{2-6}
{\tt Pick-Largest}	&95\% 		&62				&17			&242			&112\\	
\cline{2-6}
				&98\%		&81				&19			&285			&132	\\
\hline
\end{tabular}
\caption{Comparison of {\tt Remove-Smallest} and {\tt Pick-Largest} at high coverage ratios}
\end{table}

Another important aspect of performance is the efficiency of the algorithm
in terms of speed and memory use.
Table~IV 
shows the time it took for them to decompose $M$ (of {\tt Mushroom}) into $U$ and $V$
and the amount of memory used.
\begin{table}[h]
\begin{center}
\begin{tabular}{|c|c|c|c|c|}
\hline
		&{\tt GreConD}~\cite{belohlavek2013}		&{\tt GreEss}~\cite{belohlavek2013}	&{\tt Remove-S.}	&{\tt Pick-L.}\\
\hline\hline
Time 	&18min 5.7s		&12.47s		&7.39s			&10.71s\\
\hline
Memory 	&97MB			&2MB& 2.25MB				&1.72MB\\
\hline
\end{tabular}
\end{center}
\caption{Performance comparison}
\end{table}\label{tbl:preference1}
Belohlavek, et al.~\cite{belohlavek2010} carried out extensive tests of their algorithms {\tt GreConD} and {\tt GreEss},
which can be used for exact BMD, 
on {\tt Mushroom},
and measured the time and memory requirement.
Their data for exact BMD are given in Tables~IV. 
We should mention that the platforms we used to produce our results are different from theirs,
as shown in Table~V.
Probably it is safe to say that there is not a huge difference between the two.
Unfortunately, we do not have similar data for other algorithms,
since they are not published.

\begin{table}[h]\label{tbl:belohlavek}
\begin{center}
\small
\begin{tabular}{|c|p{2.5cm}|p{3.5cm}|}
\hline
		&Ours 					&Belohlavek et al.'s \cite{belohlavek2010}\\
\hline
CPU		&AMD Athlon X2 350 Dual Core Processor (3.5GHz)&INTEL Xcon 4 (3.2GHz)\\
\hline
Memory 	&4GB (1.6GHz)				&1GB\\
\hline
OS		&Windows 7 Professional		&Not mentioned in \cite{belohlavek2010}\\
\hline
Program	&{\tt Matlab} Version R2012b	&{\tt Matlab} (partially hand-coded in C)\\
\hline
\end{tabular}
\normalsize
\end{center}
\caption{Running platforms}
\end{table}

\section{Application to Educational Data Mining}\label{sec:application}
Educational data mining has been attracting increasing interest in recent years.
It aims to discover students' mastery of knowledge,
or skills which are itemized as {\em attributes}.
In the widely studied {\em Rule Space Model}~\cite{tatsuoka2009} in cognitive assessment in education,
a Boolean matrix, named the {\em Q-matrix}, is used to represent hypothetical sets of attributes
which would be needed to answer the test items correctly.
To explain the relevance of exact BMD to the educational
{\em Q-matrix theory} developed by Tatsuoka~\cite{tatsuoka2002},
let us introduce new symbols for matrices.

\textbf{Attribute or skill set:} 
We assume that the students' {\em knowledge} can be represented by the {\em knowledge state matrix}
$A= [a_{ij}] \in \{0,1\}^{m\times k}$,
where  $a_{ij}=1$ (resp. $a_{ij}=0$) indicates that the $i^{th}$ student possesses
 (resp. does not possess) knowledge represented by the $j^{th}$ attribute.
For $i=1,2,\ldots,m$, the {\em knowledge state} of student $i$ is represented by a row vector
\[
\bm{a}_i=[a_{i1},a_{i2},\ldots,a_{ik}]. 
\]

\textbf{$Q$-matrix}: 
It is denoted by  $Q= [q_{ij}]\in \{0,1\}^{n\times k}$,
where  $q_{ij}=1$ (resp. $q_{ij}=0$) indicates that answering test item $i$
correctly requires (resp. does not require) knowing or understanding attribute (={\em concept}) $j$.
Define a row vector by
\[
\bm{q}_i=[q_{i1},q_{i2},\ldots,q_{ik}].
\]

\textbf{Response matrix}: 
Given $m$ students and $n$ test items,
the test results can be represented by a matrix $R \in \{0,1\}^{m \times n}$,
where $R[i,j]=1$ (resp. $R[i,j]=0$) indicates that the $i^{th}$ student solved
the $j^{th}$ test item correctly (resp. incorrectly).
Theoretically, 
student $i$ should be able to answer question $j$ if $\bm{a}_i \geq \bm{q}_j$
or $\overline{\bm{a}_i}\circ\bm{q}_j =0$.
We thus define the {\em ideal item response} $R[i,j]$ by
\begin{eqnarray}
R[i,j]=\begin{cases}
1 & \bm{a}_i\geq \bm{q}_j\\
0 &otherwise
\end{cases}
\end{eqnarray}

If both $Q$ and $A$ were known,
then the students' test performance,
called the {\em ideal item response pattern}~\cite{tatsuoka2009},
could be theoretically predicted.
The following result was announced in \cite{sun2014} without
proof.
Here we provide a simple but formal proof.

\begin{theorem}\label{thm:MUV}
The ideal item response matrix $R$, the knowledge state matrix $A$
and the Q-matrix $Q$ are related as follows:
\begin{equation}
R=\overline{\overline{A}\circ Q^T}.
\end{equation}
\end{theorem}

\begin{proof}
The fact that student $i$ has enough knowledge to answer question $j$
can be represented by $\bm{a}_i \geq \bm{q}_j$, which is equivalent to
$\overline{\bm{a}}_i \circ\bm{q}_j^T =0$ hence $\overline{\overline{\bm{a}_i}\circ\bm{q}_j^T} =1$
by Proposition~\ref{prop:dominates}.
If he/she doesn't,
i.e., $\bm{a}_i \not\geq \bm{q}_j$,
on the other hand,
then $\overline{\bm{a}}_i \circ\bm{q}_j^T=1$, 
and $\overline{\overline{\bm{a}_i}\circ\bm{q}_j^T} =0$.
\end{proof}

If $R$ is given but the underlying matrices $Q$ and $A$ are unknown,
we want to mine $Q$ and $A$ out of $R$.
Thanks to Theorem~\ref{thm:MUV},
by finding decomposition $\overline{R}=\overline{A}\circ Q^T$,
we can learn students' knowledge state matrix $A$ and the Q-matrix $Q$
from the test responses in $R$.
We simply set $M=\overline{R}$, $U=\overline{A}$, and $V=Q^T$,
and decompose $M$.
Thus the Q-matrix learning problem can be transformed into {\em exact}
(i.e., not approximate) Boolean matrix decomposition problem.
Here we assume that $R$ has no ``noise,'' namely it correctly represents the students'
knowledge, and mine $Q$ and $A$ from it.
Clearly, the set of collected test responses, $\cal R$, is likely to be ``noisy,'' 
because students may be able to guess correct answers by luck,
or may make silly mistakes (called ``slips''~\cite{tatsuoka2009}). 
Therefore, the discovered factors $A'$ and $Q'$ of $\cal R$ are just approximations to the true $A$ and $Q$.
This problem is a main issue in {\em Rule space}
model~\cite{liu2012,sun2014,tatsuoka2009,tatsuoka2002,zhang2013},
but is beyond the scope of this paper.

\begin{example}
Here we use the dataset of Example 3.9 in \cite{tatsuoka2009}.
Table~VI 
shows the ideal item response pattern matrix $R$ for $m=12$ students and $n=11$ test items,
\begin{table}[htb]\label{tbl:tatsuoka1}
 \centering
 \begin{tabular}{c|ccccccccccrr}
                &1 &2 &3 &4 &5 &6 &7 &8 &9 &10 &\hspace{-3mm}11\\
  \midrule
1     &1     &1     &1     &1     &1     &1     &1     &1     &1     &1     &1\\
2     &1     &1     &1     &0     &1     &0     &1     &0     &0     &0     &0\\
3     &1     &1     &0     &1     &0     &1     &0     &1     &0     &0     &0\\
4     &1     &1     &0     &0     &0     &0     &0     &0     &0     &0     &0\\
5     &1     &0     &1     &1     &1     &1     &0     &0     &1     &1     &0\\
6     &1     &0     &1     &0     &1     &0     &0     &0     &0     &0     &0\\
7     &1     &0     &0     &1     &0     &1     &0     &0     &0     &0     &0\\
8     &1     &0     &0     &0     &0     &0     &0     &0     &0     &0     &0\\
9     &0     &0     &1     &1     &0     &0     &0     &0     &1     &0     &0\\
10     &0     &0     &1     &0     &0     &0     &0     &0     &0     &0     &0\\
11     &0     &0     &1     &0     &0     &0     &0     &0     &0     &0     &0\\
12     &0     &0     &0     &0     &0     &0     &0     &0     &0     &0     &0\\
\end{tabular}
\caption{The ideal item response matrix $R^{12\times 11}$ \cite{tatsuoka2009}.}
\end{table}
while Table~VII shows the matrices $A$ and $Q$ (each with $k=4$ attributes).
In \cite{tatsuoka2009}, they constructed $R$ from the given $A$ and $Q$.
Here, taking $R$ as the input, Algorithms {\tt Remove-Smallest} and {\tt Pick-Largest}
were able to recover $A$ and $Q$.
\begin{table}[htb]\label{tbl:tatsuoka2}
 \centering
$$
A=\left(
\begin{array}{llll}
1     &1     &1     &1\\
1     &1     &1     &0\\
1     &1     &0     &1\\
1     &1     &0     &0\\
1     &0     &1     &1\\
1     &0     &1     &0\\
1     &0     &0     &1\\
1     &0     &0     &0\\
0     &0^*  &1     &1\\
0     &0^*  &1     &0\\
0     &0^*  &1     &0\\
0     &0^*  &0     &0
\end{array}
\right);~~
Q=\left(
\begin{array}{cccc}
1     &0     &0     &0\\
1     &1     &0     &0\\
0     &0     &1     &0\\
0     &0     &0     &1\\
1     &0     &1     &0\\
1     &0     &0     &1\\
1     &1     &1     &0\\
1     &1     &0     &1\\
0     &0     &1     &1\\
1     &0     &1     &1\\
1     &1     &1     &1
\end{array}
\right)
$$
\caption{Knowledge state matrix $A^{12\times 4}$ and Q-matrix $Q^{11\times 4}$.}
\end{table}

\begin{comment}
In the above example,
note that $Q[:,1]$ dominates column $Q[:,2]$.
This means that any test item that tests concept 2 automatically
tests concept 1,
in other words, attribute 1 is a prerequisite for concept 2~{\rm \cite{tatsuoka2009}}.
Students 9 to 12 have not mastered concept 1,
which are tested in test items 1,2, 5$\sim$8, and 10$\sim$11.
Thus $R[s,1]=0$ (they cannot answer test items testing concept 1)
for $s=$ 9$\sim$12.
As for any test items that has a 0 in both columns 1 and 2 of $Q$,
i.e., $Q[3,:]$, $Q[4,:]$, and $Q[9,:]$,
students 9$\sim$12 (who haven't mastered concept 1)
cannot pass test items testing concept 2.
Therefore, $A[s,2]=0$ for $s=9\sim$12.
However, mathematically, setting $A[s,2]=1$ for $s=9\sim$12
still satisfies $\overline{R}=\overline{A}\circ Q^T$.
See the entries $0^*$ in matrix $A$ in Table~VII.
\qed
\end{comment}

In general, we can prove the following.
\begin{lemma}\label{lem:asterisk}
Suppose that $Q[:,i]$ dominates column $Q[:,j]$.
Then $[A[s,i]=0] \Rightarrow [A[s,j]=0]$. 
\qed
\end{lemma}

The input to our algorithms is just $\overline{R}$,
and the complemented knowledge state matrix $\overline{A}$ is an output.
Algorithm {\tt Pick-Largest} computes the values of $\delta_i$ in each round,
whose maxima are shown in Table~VIII.  

\begin{table}[h]\label{tbl:preference}
\begin{center}
\begin{tabular}{|c||c|c|c|c|c|}
\hline
Round	&1	&2	&3	&4	&5\\
\hline
$\max_i \{\delta_i\}$	&32	&30	&19	&9	&0\\
\hline
argmax$_i \{\delta_i\}$&1	&4	&3	&2	&$5\sim 11$\\
\hline
\end{tabular}
\end{center}
\medskip
\caption{The attribute picked in each round of {\tt Pick-Largest}.}
\end{table}

Algorithm {\tt Remove-Smallest} removes attributes in the increasing order of $\sigma_i$,
provided the product remains the same, i.e., $\overline{R}$.
For this example,
both algorithms decompose $\overline{R}$ into factor matrices with the common dimension ($k=4$),
which equals the dimension of the original factors \cite{tatsuoka2009}.
\qed

\end{example}

\section{Conclusions and Discussions}\label{sec:conclusion}
Given any Boolean matrix $M$, 
we first proved that $J=\overline{\overline{M}^T\circ M} \in \{0,1\}^{n\times n}$ is the ``maximal''
matrix satisfying $M=M\circ J^T$,
in the sense that if any 0 element in $J$ is changed to a 1
then this equality no longer holds.
Based on this formula,
we then presented two heuristic algorithms to find an exact decomposition $M=U\circ V$
such that $U$ consists of a ``small'' subset of the columns of $M$.
{\em Exact} BMD is aesthetically pleasing and intellectually satisfying, 
and we believe that it will find useful applications in the future.
In the present day data mining applications, however,
it may not be necessary or very important.

So we also showed that  our algorithms can be used for approximate BMD,
namely to find a product $U\circ V$ that covers most of the 1's and no 0's in $M$.
This is sometimes called ``from-below'' approximation~\cite{belohlavek2013}.
We ran our algorithms on several real examples,
which are often used as benchmarks.
On these particular datasets, our algorithms perform rather well,
compared with the known algorithms proposed in~\cite{belohlavek2013,belohlavek2010,xiang2011,geerts2004},
despite the column use condition that we impose, but the others do not.
Clearly, more extensive tests are called for to arrive at any definite conclusions.

Although we have concentrated on the elimination of column dominance,
it is possible that a given matrix $M$ has more row dominance than column dominance.
In any case, it would be worthwhile to apply our algorithms to both $M$ and $M^T$,
and pick the result with the smaller factor matrix dimensions.
There may be situations where a decomposition of $M=A\circ B$ is already known,
but it is desired to reduce the number of attributes (columns) in $A$.
In such a case, we can apply our algorithms to decompose $A$ as $A=U\circ V$.
We then have $M=U\circ (V\circ B)$ such that $U$ consists of a subset of the columns of $A$.

We are working on a promising BMD algorithm without column use condition,
which is founded on some mathematical formulae proved in this paper.
For dataset {\tt Mushroom}, it achieves 100\% coverage for $k= 101$,
which is better than any algorithm in Table~II, as for as exact BMD is concerned.

As a final remark, from Proposition~\ref{prop:expand} 
there is a lot of parallelism in matrix product computation.
This implies that if the given matrix $M$ is very large,
our algorithms are amenable to the {\em map-reduce} technique~\cite{rajaraman2014}.

Finally, as mentioned before, we have not examined the {\em approximation ratio} of
our heuristic algorithms relative to the optimum.
We leave it as future work.

\section*{Acknowledgement}
We thank Martin Trnecka of Palack\'{y} University, Czech Republic,
 for providing us with the testbed datasets tailored as inputs to {\tt Matlab}.


\end{document}